\documentclass[12pt]{article}
\usepackage{mathrsfs}
\usepackage{graphicx}
\usepackage{enumerate}
\makeatother
\usepackage{amsfonts,amsmath,amssymb}
\usepackage{epstopdf}
\usepackage{caption}
\usepackage{subcaption}
\usepackage{amsmath}

\title{}\date{}
\title{Existence of Coupled Optical Vortex Solitons Propagating in a Quadratic Nonlinear Medium}
\author{Luciano Medina
\footnote{lmedina@nyu.edu}\\Shokan, New York 12481, USA}

\def\XXint#1#2#3{{\setbox0=\hbox{$#1{#2#3}{\int}$}
 \vcenter{\hbox{$#2#3$}}\kern-.5\wd0}}

\newtheorem{oldtheorem}{Theorem}[section]
\newtheorem{oldassertion}[oldtheorem]{Assertion}
\newtheorem{oldproposition}[oldtheorem]{Proposition}
\newtheorem{oldremark}[oldtheorem]{Remark}
\newtheorem{oldlemma}[oldtheorem]{Lemma}
\newtheorem{olddefinition}[oldtheorem]{Definition}
\newtheorem{oldclaim}[oldtheorem]{Claim}
\newtheorem{oldcorollary}[oldtheorem]{Corollary}

\newenvironment{theorem}{\begin{oldtheorem}$\!\!\!${\bf.}}{\end{oldtheorem}}

\newenvironment{lemma}{\begin{oldlemma}$\!\!\!${\bf.}}{\end{oldlemma}}

\newbox\qedbox

\newenvironment{proof}{\smallskip\noindent{\bf Proof.}\hskip \labelsep}%
                        {\hfill\penalty10000\copy\qedbox\par\medskip}

\newenvironment{acknowledgements}{\smallskip\noindent{\bf Acknowledgements.}
        \hskip\labelsep}{}

\providecommand{\keywords}[1]{\textbf{\textbf{Keywords---}} #1}
\begin{document}
\maketitle
\begin{abstract}
We consider the coupled propagation of an optical field and its second harmonic in a quadratic nonlinear medium governed by a coupled system of Schrodinger equations. We prove the existence of ring-profiled optical vortex solitons appearing as solutions to a constrained minimization problem and as solutions to a min-max problem. In the case of the constrained minimization problem solutions are shown to be positive with undetermined wave propagation constants, but in the min-max approach the wave propagation constants can be prescribed. The quadratic nonlinearity introduces some interesting properties not commonly observed in other coupled systems in the context of nonlinear optics, such as the system not accepting any semi-trivial solutions, meaning, that optical solitons cannot be observed when, say, one of the beams are off. Additionally, the second harmonic always remains positive. 
\end{abstract}

{\bf 2010 Mathematics Subject Classification.} 35J20, 35J50, 35Q55, 35Q60.

\medskip
\keywords{optical vortices, Schrodinger equations, constrained minimization, Palais-Smale condition, mountain-pass theorem.}

\section{Introduction}
Nonlinear optics, which studies the effects of solitary waves due to optical propagation in a nonlinear medium, is a very active area of investigation in both theoretical and experimental research. Its applications are abound in nonlinear science and can be found in quantum information processing, wireless communications, condensed matter physics, particle interactions, rogue waves formation, and cosmology 
\cite{ABSW, A, BKKH,BSV,KMT,RSS,Scott, SO,SL, TT, YW}.

Our interest is in the work of Skryabin and Firth \cite{SF1, SF2}, which presents the dynamics and stability of ring-profiled solitary waves propagating in a self-focusing saturable non-linear medium and in a quadratic non-linear medium. In \cite{LM,LM5} we established an existence theory in two cases where the non-linearity is of the saturable type and governed by a single nonlinear Schrodinger equation. Similar analysis for other types of non-linearities have also been considered, such as, cubic \cite{YR,GuoCaoLi} and cubic-quintic \cite{Greco}. In this study, our interest is the coupled propagation of the wave and its second harmonic governed by a coupled nonlinear system of Schrodinger equations with a quadratic non-linearity. Such coupled systems form so called soliton-induced waveguides \cite{DBF, DKT, KA, Sk, LY, LX, SK,SS} and have been studied by many mathematical analysts \cite{AC, AC2, Boyan, MMP,RuifengNan}.  In the case of two mutually incoherent optical beams propagating in a self-focusing nonlinear saturable medium \cite{SK}, optical vortex solitons exist even if one beam is off, which translates to the existence of semi-trivial solutions in the mathematical context. However, we show that the coupled system with quadratic nonlinearity does not allow any semi-trivial solutions and establish an existence theory for this system.

To this end, consider the field envelopes $E_1$ and $E_2$ of an optical field and its second harmonic described by the dimensionless coupled nonlinear Schrodinger equations \cite{SF1,SF2}:
\begin{align}
    i\partial_zE_1+\dfrac{1}{2}\vec{\nabla}^2_{\perp}E_1+E^*_1E_2&=0\label{Schrodinger1}\\
    i\partial_zE_2+\dfrac{1}{4}\vec{\nabla}^2_{\perp}E_2+\dfrac{1}{2}E^2_1&=\beta E_2,\label{Schrodinger2}
\end{align}
where $\beta$ is the phase mismatch parameter. The wave propagation is in the longitudinal $z$-direction over the transverse plane of coordinates $(x,y)$ perpendicular to the $z$-axis. $\vec{\nabla}_{\perp}^2$ is the Laplace operator over the transverse plane of coordinates. Optical vortex solitons are localized solutions of \eqref{Schrodinger1}-\eqref{Schrodinger2}, which do not change their intensity profile during propagation and have a phase singularity at its center. They are described under the ansatz
\begin{align}
    E_m=A_m(r)e^{im(l\theta+\kappa z)},\quad m =1,2,\label{ansatz}
\end{align}
where $A_1$ and $A_2$ are real-valued functions representing the field amplitudes, i.e., $|E_m|=A_m^2$, $m=1,2$, respectively. $l$ is the azimuthal model index or a free parameter, which we will call the vortex number and is restricted to take on integer values to ensure azimuthal periodicity. $\kappa$ is the wave propagation constant. $r=\sqrt{x^2+y^2}$ and $\theta$ is the polar angle. Under the ansatz \eqref{ansatz}, the coupled nonlinear Schrodinger equations reduce to the so called $l$-vortex system, 
\begin{align}
&A_{1,rr}+\dfrac{1}{r}A_{1,r}-\dfrac{l^2}{r^2}A_1=2(\kappa-A_2)A_1,\label{sysEq1}\\
&A_{2,rr}+\dfrac{1}{r}A_{2,r}-\dfrac{4l^2}{r^2}A_2=4(2\kappa+\beta)A_2-2A_1^2,\label{sysEq2}
\end{align}
and considered under the boundary conditions
\begin{align}
    A_1(0)=0=A_1(R),\qquad A_2(0)=0=A_2(R), \label{bdyCond}
\end{align}
where the first boundary condition, $A_m(0) = 0$, $m=1,2$, is due to the presence of the vortex core or, equivalently, the regularity of $A_m$, $m=1,2$, at $r = 0$, and the second boundary condition, $A_m(R) = 0$, $m=1,2,$ for $R > 0$ sufficiently large, represents the distance from the vortex core and
may be imposed due to beam confinement. 

An important parameter characterization of spatial vortex solitons is defined by the integral
\begin{align}
    Q(A_m)=\int_0^{2\pi}\int_0^R|E_m|rdrd\theta=2\pi\int_0^R A_m^2 rdr,\quad m=1,2,\label{energyFlux}
\end{align}
where we refer to $Q$ as the energy flux and $Q(A_1)+2Q(A_2)$ is the total energy flux of the system. We define the energy functional 
\begin{align}
    \mathcal{E}(A_1,A_2)=\int_0^R\left(A_{1,r}^2+A_{2,r}^2+\dfrac{A_1^2}{r^2}+\dfrac{A_2^2}{r^2}+A_1^2A_2\right)rdr,
\end{align}
and refer to solutions as being of finite energy whenever $\mathcal{E}(A_1,A_2)<\infty$. 

Our specific interest is in establishing the existence of finite energy fully non-trivial exponentially decaying solution pairs, $(A_1,A_2)$, of the system \eqref{sysEq1}-\eqref{bdyCond}. By fully non-trivial solutions, we mean solutions that are not the trivial solution, $(0,0)$, and not semi-trivial, meaning, $(A_1,A_2)$, where neither $A_1\equiv 0$ or $A_2\equiv 0$. To this end, we summarize our main results with the following theorems, 
\begin{oldtheorem}
    Let the azimuthal model index, $l$, be any nonzero integer and consider the boundary value problem \eqref{sysEq1}-\eqref{bdyCond} governing the amplitude of the optical field and its second harmonic propagating along the longitudinal $z$-direction with a propagation constant $\kappa$ and a phase mismatch constant $\beta$. For $0<Q(A_1):=Q_1<2\pi |l|$ and any $Q(A_2):=Q_2>0$, there exists a solution pair $(A_1,A_2)$ with $A_m(r)>0$, $r\in(0,R)$, $m=1,2$, and $\kappa,\beta\in\mathbb{R}$ arising as Lagrange multipliers of a constrained optimization problem. 
\end{oldtheorem}

\begin{theorem}
    Let $\kappa>\max\{0,-\beta/2\}$ and $(A_1,A_2)$ be a nontrivial solution pair to  \eqref{sysEq1}-\eqref{bdyCond}.
    \begin{enumerate}[(i)]
    \item The second harmonic $A_2$ is always positive, i.e., $A_2(r)>0$ for all $r\in(0,R)$.
    \item There exist no semi-trivial solutions.
    \item Let $M_m:=\max\limits_{r\in(0,R)}|A_m(r)|$, $m=1,2$. The global maximum of the second harmonic satisfies
    \begin{align}
        \dfrac{l^2}{2R^2}+\kappa< M_2<\dfrac{M_1^2}{\dfrac{2l^2}{R^2}+2(2\kappa+\beta)}.
    \end{align}
    \item There holds the exponential decay estimate
    \begin{align}
A^2_1(r)\leq C\exp(-\sqrt{2\kappa}r)\quad \text{and}\quad A^2_2(r)\leq C\exp(-\sqrt{2\kappa+\beta}r),
\end{align}
for $r$ sufficiently large and $C>0$ a constant dependent on $\kappa$ and $\beta$ only.
    \end{enumerate}
\end{theorem}

\begin{theorem}
    For any $\kappa>\max\{0,-\beta/2\}$, $|l|\geq 1$, and $R>0$, the coupled system \eqref{sysEq1}-\eqref{sysEq2} satisfying the boundary conditions \eqref{bdyCond} has a fully non-trivial solution defined over $[0,R]$. Moreover, such solution is a saddle point of an indefinite action functional and appears as a result of a min-max approach. 
\end{theorem}

The remainder of this paper is as follows. In section 2 we establish Theorem 1 via a constrained optimization approach. Then various results concerning the solutions of the coupled system, such as, the non-existence of semi-trivial solutions, the second harmonic being positive and the exponential decay of the solutions are established in Section 3. In section 4, we conclude with the prove of Theorem 1.3.

\section{Existence of Vortices via a Constrained Minimization Problem}
\setcounter{equation}{0}
Here we prove that the boundary value problem \eqref{sysEq1}-\eqref{bdyCond} may be solved via a constrained optimization problem. To achieve this, first consider the action functional 
\begin{equation}
    I(A_1,A_2)=\dfrac{1}{2}\int_0^R\left(A_{1,r}^2+\dfrac{1}{2}A_{2,r}^2+\dfrac{l^2}{r^2}A_1^2+\dfrac{2l^2}{r^2}A_2^2-2A_1^2A_2\right)rdr,\label{Func1}
\end{equation}
and energy flux constraint functionals 
\begin{equation}
    Q(A_1)=2\pi\int_0^RA_1^2rdr=Q_1>0\quad Q(A_2)=2\pi\int_0^RA_2^2rdr=Q_2>0.\label{constFunc}
\end{equation}
It then suffices to prove the existence of a solution to the constrained optimization problem, 
\begin{equation}
    \text{min}\{I(A_1,A_2):(A_1,A_2)\in\mathcal{C},\quad Q(A_1)=Q_1>0, \quad Q(A_2)=Q_2>0\},\label{OptProb}
\end{equation}
defined over the admissible class 
\begin{equation}
    \mathcal{C}=\left\{A_1,A_2\text{ are absolutely continuous on $[0,R]$, satisfy \eqref{bdyCond}, $\mathcal{E}(A_1,A_2)<\infty$}\right\}.
\end{equation}

Applying the basic inequality $2ab\leq a^2+b^2$, for all $a,b\in\mathbb{R}$, to \eqref{Func1} and using the constraint, $Q(A_2)=Q_2>0$, we have
\begin{align}
    I(A_1,A_2)&\geq\dfrac{1}{2}\int\left(A_{1,r}^2+\dfrac{1}{2}A_{2,r}^2+\dfrac{l^2}{r^2}A_1^2+\dfrac{2l^2}{r^2}A_2^2-A_1^4\right)rdr-\dfrac{Q_2}{4\pi}.\label{lowerBnd_1}
\end{align}
For any function satisfying $A(0)=0$, the Cauchy-Schwartz inequality gives
\begin{equation}
    A^2(r)=\int_0^r2A(\rho)A_r(\rho)d\rho\leq 2\left(\int_0^r\rho A_{\rho}^2(\rho)d\rho\right)^{1/2}\left(\int_0^r\dfrac{1}{\rho} A^2(\rho)d\rho\right)^{1/2}.
\end{equation}
Multiplying by $rA^2$ and integrating, we arrive at 
\begin{equation}
    \int_0^RrA^4dr\leq 2\left(\int_0^RrA^2dr\right)\left(\int_0^R rA_r^2dr\right)^{1/2}\left(\int_0^R\dfrac{1}{r} A^2dr\right)^{1/2}.\label{bndOnA^4}
\end{equation}
Using Cauchy's inequality with $\epsilon$ \cite{Evans}, the constraint $Q(A_1)=Q_1>0$, and \eqref{bndOnA^4}, we get
\begin{align}
    \int_0^RrA_1^4dr&\leq \dfrac{Q_1}{\pi}\left(\int_0^R rA_{1,r}^2dr\right)^{1/2}\left(\int_0^R\dfrac{1}{r} A_1^2dr\right)^{1/2}\\
    &\leq \epsilon\int_0^R rA_{1,r}^2dr+\dfrac{Q_1^2}{4\pi^2\epsilon}\int_0^R\dfrac{1}{r} A_1^2dr.\label{bndOnA_1^4}
\end{align}
Consequently, using \eqref{bndOnA_1^4} in \eqref{lowerBnd_1}, we obtain
\begin{align}
    I(A_1,A_2)&\geq\dfrac{1}{2}(1-\epsilon)\int_0^R rA_{1,r}^2dr+\dfrac{1}{2}\left(l^2-\dfrac{Q_1^2}{4\pi^2\epsilon}\right)\int_0^R\dfrac{1}{r}A_1^2dr\\
    &+\dfrac{1}{4}\int_0^R rA_{2,r}^2dr+2l^2\int_0^R\dfrac{1}{r}A_2^2dr-\dfrac{Q_2}{2\pi}.\nonumber
\end{align}
We can now choose $\epsilon>0$ such that the inequalities 
\begin{equation}
    1-\epsilon>0\qquad\text{and}\qquad l^2-\dfrac{Q_1^2}{4\pi^2\epsilon}>0
\end{equation}
are simultaneously satisfied, which leads to
\begin{equation}
    Q_1<2\pi|l|.
\end{equation}
Consequently, assuming $Q_1<2\pi|l|$, we then get the coercive lower bound
\begin{align}
    I(A_1,A_2)&\geq C_1\int_0^R rA_{1,r}^2dr+C_2\int_0^R\dfrac{A_1^2}{r}dr+\dfrac{1}{4}\int_0^R rA_{2,r}^2dr+2l^2\int_0^R\dfrac{A_2^2}{r}dr-\dfrac{Q_2}{2\pi},\label{CoerciveBnd}
\end{align}
where $C_1,C_2$ are positive constants depending on $\epsilon, l,Q_1$, but independent of $A_1$ and $A_2$.

Consequently, we can now choose a minimizing sequence of \eqref{OptProb}, say, $\{(A_{1,n},A_{2,n})\}_{n=1}^{\infty}$, and applying the coercive lower bound \eqref{CoerciveBnd} achieve the upper bound
\begin{align}
    \int_0^R \left(r[A_{1,n}]_r^2+\dfrac{A_{1,n}^2}{r}\right)dr+\int_0^R \left(r[A_{2,n}]_r^2+\dfrac{A_{2,n}^2}{r}\right)dr\leq C,\label{seqBnd}
\end{align}
where $C>0$ is a constant independent of $m$ and we use the notation $[A]_r=\dfrac{dA}{dr}$. 

Using the fact that the distributional derivative \cite{GT} satisfies $||A|_r|\leq |A_r|$ and the simple inequality $-b\geq-|b|$, $b\in\mathbb{R}$, we observe that $I(|A_{1,n}|,|A_{2,n}|)\leq I(A_{1,n},A_{2,n})$. We also have that $Q$ is an even functional so that $Q(A)=Q(|A|)$. Therefore, we may modified the sequence $\{(A_{1,n},A_{2,n})\}_{n=1}^{\infty}$ such that each $A_{1,n}$ and $A_{2,n}$ are non-negative, i.e., $A_{1,n}(r)\geq 0$ and $A_{2,n}(r)\geq 0$ for all $r\in(0,R)$. So we assume $\{(A_{1,n},A_{2,n})\}_{n=1}^{\infty}$ is a sequence of non-negative valued functions. We can also view each of these functions as radially symmetric functions defined over the disk $D_R=\left\{(x,y)\in\mathbb{R}^2|x^2+y^2\leq R^2\right\}$ and vanishing on the boundary of $D_R$. Moreover, from \eqref{seqBnd} and the inequality
\begin{align}
    \int_0^R rA^2dr\leq R^2\int_0^R\dfrac{1}{r}A^2dr,
\end{align}
we see that $\{(A_{1,n},A_{2,n})\}_{n=1}^{\infty}$ is bounded under the radially symmetric norm defined by 
\begin{align}
 ||(A_1,A_2)||^2=\int_0^R \left(A_{1,r}^2+A_{2,r}^2+A_1^2+A_{2}^2\right)rdr,
\end{align}
over the product space $H=W^{1,2}_0(D_R)\times W^{1,2}_0(D_R)$ with the induced component-wise operations of the standard Sobolev space $W^{1,2}_0(D_R)$. Without loss of generality, we may assume that the sequence $(A_{1,n},A_{2,n})\rightharpoonup (A_1,A_2)$ converges weakly in $H$ as $n\rightarrow\infty$. From the compact embedding $W^{1,2}(D_R)\hookrightarrow L^p(D_R)$, for $p\geq 1$, we get the strong convergence  $(A_{1,n},A_{2,n})\rightarrow (A_1,A_2)$ in $L^p(D_R)\times L^p(D_R)$ as $n\rightarrow\infty$. It also follows that $(A_1,A_2)$ is a pair of radially symmetric functions and satisfy $A_1(R)=0=A_2(R)$. 

We still need to show that $A_1(0)=0=A_2(0)$. To this end, we observe that the sequence $\{(A_{1,n},A_{2,n})\}_{n=1}^{\infty}$ is bounded over the product space $W^{1,2}(\epsilon,R)\times W^{1,2}(\epsilon,R)$ for any $\epsilon>0$. From the compact embedding of $W^{1,2}(\epsilon,R)\hookrightarrow C[\epsilon,R]$ we get the uniform convergence of $(A_{1,n},A_{2,n})\rightarrow (A_1,A_2)$ defined over $[\epsilon,R]\times [\epsilon,R]$ as $n\rightarrow\infty$. For any pair $r_1,r_2\in(0,R)$, with $r_1<r_2$, and $m=1,2$, we have 
\begin{align}
|A_{m,n}^2(r_2)-A_{m,n}^2(r_1)|&=\left|\int_{r_1}^{r_2}\dfrac{d}{dr}\left(A_{m,n}(r)\right)^2dr\right|\label{uniformConv}\\
&\leq\int_{r_1}^{r_2}2|A_{m,n}(r)||[A_{m,n}]_r(r)|dr\nonumber\\
&\leq 2 \left(\int_{r_1}^{r_2}r[A_{m,n}]_r^2(r)dr\right)^{1/2}\left(\int_{r_1}^{r_2}\dfrac{A_{m,n}^2(r)}{r}dr\right)^{1/2}\nonumber\\
&\leq 2 C^{1/2}\left(\int_{r_1}^{r_2}\dfrac{A_{m,n}^2(r)}{r}dr\right)^{1/2},\nonumber
\end{align}
where the last inequality and $C$ follows from \eqref{seqBnd}. Taking the limit as $n\rightarrow\infty$ in \eqref{uniformConv}, we get
\begin{align}
|A_{m}^2(r_2)-A_{m}^2(r_1)|&\leq 2 C^{1/2}\left(\int_{r_1}^{r_2}\dfrac{A_{m}^2(r)}{r}dr\right)^{1/2}.\label{2bnd}
\end{align}
From \eqref{seqBnd} and Fatou's lemma, we have
\begin{align}
\int_0^R r[A_{m}]_r^2dr\leq \liminf_{n\rightarrow\infty}\int_0^R r[A_{m,n}]_r^2dr,\label{Fatou1}\\
\int_0^R \frac{A_{m}^2}{r}dr\leq \liminf_{n\rightarrow\infty}\int_0^R \frac{A_{m,n}^2}{r}dr\label{Fatou2}.
\end{align}
Consequently, and in view of \eqref{seqBnd}, $\frac{1}{r}A_{m}^2\in L(0,R)$ for each $m=1,2$. Hence, as $r_1,r_2\rightarrow 0$, the right-hand side of \eqref{2bnd} goes to zero and gives the existence of the limits
\begin{align}
\xi_{1}=\lim\limits_{r\rightarrow 0}A_1^2(r),\quad \text{and}\quad \xi_{2}=\lim\limits_{r\rightarrow 0}A_2^2(r).
\end{align}
Since $\frac{1}{r}A_{1}^2,\frac{1}{r}A_{2}^2\in L(0,R)$ we have $\xi_{1}=0$ and $\xi_{2}=0$. Therefore, $A_{1}(0)=0=A_2(0)$ as desired. 

Summarizing our results, we have that $(A_1,A_2)$ obtained as the limit of the minimizing sequence $\{(A_{1,n},A_{2,n})\}_{n=1}^{\infty}$ for the problem \eqref{OptProb} satisfies $A_m(0)=0=A_m(R)$, $A_m(r)\geq 0$, for all $r\in[0,R]$, $m=1,2$, $\mathcal{E}(A_1,A_2)<\infty$, and 
\begin{align}
I(A_1,A_2)\leq\liminf\limits_{n\rightarrow\infty}I(A_{1,n},A_{2,n}),\quad Q(A_m)=\lim\limits_{n\rightarrow\infty}Q(A_{m,n})=Q_m,\quad m=1,2.
\end{align}
Therefore, $(A_1,A_2)$ is a solution to the coupled system \eqref{sysEq1}-\eqref{bdyCond} in which parameters $\kappa,\beta\in\mathbb{R}$ appearing as Lagrange multipliers of the constrained minimization problem \eqref{OptProb}. 

Moreover, such a solution $(A_1,A_2)$ satisfies $A_m(r)>0$, $m=1,2$. If this was not the case,  there would exists a point $r_0\in(0,R)$ such that $A_m(r_0)=0$, $m=1,2$, then $\frac{d}{dr}A_m(r_0)=0$ since $r_0$ would be a minimum of $A_m(r)$. Then, by the uniqueness theorem for the initial value problem of ordinary differential equations, we would have $A_m(r)=0$ for all $r\in(0,R)$, which is a contradiction to the energy flux constraint $Q(A_m)=Q_m>0$, $m=1,2$. 

\section{The second harmonic, non-existence of semi-trivial solutions, and miscellaneous results}
\setcounter{equation}{0}
In this section we prove that the second harmonic is always positive, that the system \eqref{sysEq1}-\eqref{bdyCond} possesses no semi-trivial solutions, establish estimates on the global max of the waves, and obtain an exponential decay estimate. With this goal in mind, we provide each result as an independent lemma which together form Theorem 1.2.

\begin{lemma}\label{PositivityOfA_2}
    Let $\kappa>-\beta/2$ and $(A_1,A_2)$ be a non-trivial solution pair of  \eqref{sysEq1}-\eqref{bdyCond}. The second harmonic $A_2$ is always positive, i.e., $A_2(r)>0$ for all $r\in(0,R)$.
\end{lemma}
\begin{proof}
    Let $A_2$ be a solution of  \eqref{sysEq1}-\eqref{bdyCond} and suppose there is a point $r_0\in(0,R)$ such that $A_2(r_0)<0$. Then there must exist a second point $r_1\in(0,R)$ such that $A_2(r_1)<0$, $A_{2,r}(r_1)=0$, and $A_{2,rr}(r_1)>0$. Substituting into equation \eqref{sysEq2}, we get
    \begin{align}
        0<A_{2,rr}(r_1)=\left(\dfrac{4l^2}{r_1^2}+4(2\kappa+\beta)\right)A_2(r_1)-2A_2^2(r_1)<0,
    \end{align}
    a contradiction. Consequently, there is no $r_0\in(0,R)$ such that $A_2(r_0)<0$.$\qquad\square$
\end{proof}

\begin{lemma}\label{NonSemiTrivialSolns}
For $\kappa>-\beta/2$ and  $(A_1, A_2)$ be a non-trivial solution pair of \eqref{sysEq1}-\eqref{bdyCond}. Then $(A_1,A_2)$ is not a semi-trivial solution, i.e., neither $A_1\equiv0$ or $A_2\equiv 0$.
\end{lemma}
\begin{proof}
Let $(A_1,A_2)$ be a nontrivial solution pair to \eqref{sysEq1}-\eqref{bdyCond}, i.e., $(A_1,A_2)\not\equiv (0,0)$. Suppose $A_2\equiv 0$ and substitute into \eqref{sysEq2}, to get $0=A_1^2(r)$ for all $r\in[0,R]$. Hence, $A_1\equiv 0$, a contradiction to the non-triviality of the solution.  

Suppose $A_1\equiv 0$ and substitute into \eqref{sysEq2} to get
\begin{align}
    A_{2,rr}+\dfrac{1}{r}A_{2,r}=\left(\dfrac{4l^2}{r^2}+4(2\kappa +\beta)\right)A_{2},\label{sysEq2_mod}
\end{align}
and note that if $A_2$ is a solution, then $-A_2$ is also a solution. Thus there is an $r_0\in(0,R)$ such that $A_2(r_0)>0$, $A_{2,r}(r_0)=0$, and $A_{2,rr}(r_0)<0$. Then \eqref{sysEq2_mod}, gives 
\begin{align}
    0>A_{2,rr}(r_0)=\left(\dfrac{4l^2}{r_0^2}+4(2\kappa +\beta)\right)A_{2}(r_0)>0, 
\end{align}
a contradiction. Consequently, $A_1\not\equiv 0$. $\qquad\square$
\end{proof}
\begin{lemma}
Let $\kappa>\max\{0,-\beta/2\}$, $(A_1,A_2)$ be a nontrivial solution pair of \eqref{sysEq1}-\eqref{bdyCond}, and $M_m:=\max\limits_{r\in(0,R)}|A_m(r)|$, $m=1,2$. Then 
    \begin{align}
        \dfrac{l^2}{2R^2}+\kappa< M_2<\dfrac{M_1^2}{\dfrac{2l^2}{R^2}+2(2\kappa+\beta)}.
    \end{align}
    \end{lemma}
    \begin{proof}
        Suppose $(A_1,A_2)$ is a nontrivial solution pair to \eqref{sysEq1}-\eqref{bdyCond}. Let $M_m:=\max\limits_{r\in(0,R)}A_m(r)$, $m=1,2$. By Lemma 3.1, $A_2(r)>0$ for all $r\in(0,R)$. Then there is an $r_2\in(0,R)$ such that $M_2=A_2(r_2)>0$, $A_{2,r}(r_2)=0$, and $A_{2,rr}(r_2)<0$. Inserting into equation \eqref{sysEq2} gives
        \begin{align}
            0>A_{2,rr}(r_2)=\left(\dfrac{4l^2}{r_2^2}+4(2\kappa+\beta))\right)M_2-2A_1^2(r_2)
        \end{align}
        or equivalently
        \begin{align}
            M_2<\dfrac{A_1^2(r_2)}{\dfrac{2l^2}{R^2}+2(2\kappa+\beta)}\leq \dfrac{M_1^2}{\dfrac{2l^2}{R^2}+2(2\kappa+\beta)}.
        \end{align}
    
        Suppose there is an $r_0\in(0,R)$ such that $A_1(r_0)>0$, then there will also be an $r_1\in(0,R)$ such that $A_{1,rr}(r_1)<0$, $A_{1,r}(r_1)=0$, and $A_1(r_1)>0$. Substituting this $r_1$ into \eqref{sysEq1}, then gives 
        \begin{align}
            0>A_{1,rr}(r_1)=\left(\dfrac{l^2}{r_1^2}+2(\kappa-A_2(r_1))\right)A_1(r_1)
        \end{align}
        and it follows that
        \begin{align}
            \dfrac{l^2}{2R^2}+\kappa<A_2(r_1)\leq M_2.
        \end{align}
        On the other hand, if there is an $r_0\in(0,R)$ such that $A_1(r_0)<0$, then there also is an $r_1\in(0,R)$ such that $A_{1,rr}(r_1)>0$, $A_{1,r}(r_1)=0$, and $A_1(r_1)<0$. Consequently, substituting into \eqref{sysEq1}, gives
         \begin{align}
            0<A_{1,rr}(r_1)=\left(\dfrac{l^2}{r_1^2}+2(\kappa-A_2(r_1))\right)A_1(r_1)
        \end{align}
        and again it follows that
        \begin{align}
            \dfrac{l^2}{2R^2}+\kappa<A_2(r_1)\leq M_2.
        \end{align}
        In either case, we get the desire inequality. 
        $\qquad\square$
    \end{proof}

Localized solutions and beam confinement demand that the wave amplitude decay exponentially fast for large $R$. The following lemma shows that the exponential decay estimate follows from a direct application of the maximum principle and a suitable comparison function.

\begin{lemma}
    Let $\kappa>\max\{0,-\beta/2\}$ and $(A_1,A_2)$ be a solution to \eqref{sysEq1}-\eqref{bdyCond}. Then there is an $R_0\in(0,R)$ such that
        \begin{align}
A^2_1(r)\leq C\exp(-\sqrt{2\kappa}r)\quad \text{and}\quad A^2_2(r)\leq C\exp(-\sqrt{2\kappa+\beta}r),
\end{align}
for every $r\in(R_0,R]$ and $C>0$ a constant dependent on $\kappa$ and $\beta$ only.
\end{lemma}
\begin{proof}
    Rewrite \eqref{sysEq1} as follows,
    \begin{align}
        \Delta A_1 = A_{1,rr}+\dfrac{1}{r}A_{1,r}=\left(\dfrac{l^2}{r^2}+2(\kappa-A_2)\right)A_1.
    \end{align}
We then have
    \begin{align}
        \Delta A_1^2 \geq 2A_1\Delta A_1= 2\left(\dfrac{l^2}{r^2}+2(\kappa-A_2)\right)A_1^2\geq 4\left(\kappa-A_2\right)A_1^2.
    \end{align}
By the continuity of $A_1$ on $[0,R]$  and the boundary condition $A_1(R)=0$, for any $\epsilon_1>0$ there is an $R_1\in (0,R)$ such that    
    \begin{align}
        \Delta A_1^2 \geq 4\left(\kappa-\epsilon_1\right)A_1^2\quad\text{ for every $r\in [R_1,R]$}.
    \end{align}
    Consider the comparison function $\xi_1:[0,R]\rightarrow\mathbb{R}$ defined by
    \begin{align}
        \xi_1(r)=C_1\exp(-\sigma_1^2 r), \quad C_1,\sigma_1>0.\label{CompFunc}
    \end{align}
    For every $r\in[R_1,R]$ we have
    \begin{align}
        \Delta(A_1^2-\xi_1)\geq 4\left(\kappa-\epsilon_1\right)A_1^2-\left(\sigma^2\xi-\dfrac{\sigma_1 \xi}{r}\right)\geq 4\left(\kappa-\epsilon\right)A_1^2-\sigma_1^2\xi_1.
    \end{align}
    Let $\sigma_1^2=4(k-\epsilon_1)$, to get
    \begin{align}
        \Delta(A_1^2-\xi_1)\geq \sigma_1^2(A_1^2-\xi_1).
    \end{align}
    Now we can choose $C_1$ in \eqref{CompFunc} large enough so that $A_1^2-\xi_1\leq 0$ for $r=R_1$. Since $A_1(r)\rightarrow 0$ as $r\rightarrow R^-$, and applying the maximum principle, we conclude that $A_1^2-\xi_1\leq 0$ for all $r\in[R_1,R]$. For simplicity, we choose $\epsilon_1=\kappa/2>0$ and get our desired result
    \begin{align}
        A_1^2\leq C_1\exp(-\sqrt{2\kappa}r)\quad\text{ for every $r\in [R_1,R]$},
    \end{align}
    where $C_1>0$ and $R_1>0$ are constants depending on $\kappa$ only.
    
    In a similar manner as above, with the comparison function $\xi_2:[0,R]\rightarrow\mathbb{R}$ defined by
    \begin{align}
        \xi_2(r)=C_2\exp(-\sigma_2^2 r), \quad C_2,\sigma_2>0,\label{CompFunc2}
    \end{align}
    we can arrive at
    \begin{align}
        \Delta(A_2^2-\xi_2)\geq \sigma_2^2(A_2^2-\xi_2),
    \end{align}
    where $\sigma_2^2=4(2\kappa+\beta)$. Then, via the maximum principle and $C_2,R_2$ sufficiently large, conclude 
    \begin{align}
        A_2^2\leq C_2\exp(-2\sqrt{2\kappa+\beta}r)\quad\text{ for every $r\in [R_2,R]$},
    \end{align}
    where $C_2>0$ and $R_2>0$ are constants depending on $\kappa$ and $\beta$ only. We may then take $C=\max\{C_1,C_2\}$ and $R_0=\max\{R_1,R_2\}$. 
    $\qquad\square$
\end{proof}

\section{Saddle Point Solutions via a Mountain Pass Theorem}
\setcounter{equation}{0}
Theorem 1.1 establishes the existence of a non-trivial positive solution via a constrained optimization problem, but the parameters $\kappa$ and $\beta$ were undetermined. In this section, we prove Theorem 1.3 and show that a fully non-trivial solution exists as a saddle point of an indefinite action functional where the wave propagation constant $\kappa$ and its second harmonic $\beta$ may be prescribed on a continuous range of values. With this goal in mind, consider the action functional 
\begin{align}
    J(A_1,A_2)&=\dfrac{1}{2}\int_0^R\left(A_{1,r}^2+\dfrac{1}{2}A_{2,r}^2+\dfrac{l^2}{r^2}A_1^2+\dfrac{2l^2}{r^2}A_2^2\right)rdr\label{IndfAction}\\
    &+\int_0^R\left((\kappa-A_2)A_1^2+(2\kappa+\beta)A_2^2-A_1^2A_2\right)rdr,\nonumber
\end{align}
with $|l|\geq 1$ and 
\begin{align}
    \kappa>\max\{0,-\beta/2\}.
\end{align}
Let $H$ be the completion of the space $X=\left\{A\in \mathcal{C}^1[0,1]:A(0)=0=A(R)\right\}$ (the set of differentiable functions over $[0,R]$ which vanish at the two endpoints of the interval), with the inner product and norm,
\begin{align}
\langle A,\tilde{A}\rangle&=\int_0^R\left\{A_{r}\tilde{A}_r+\dfrac{l^2}{r^2}A\tilde{A}\right\}rdr,\quad ||A||^2=\langle A,A\rangle,
\end{align}
respectively. 
From 
\begin{align}
||A||^2=\int_0^R\left\{A_r^2+\dfrac{l^2}{r^2}A^2\right\}rdr\geq\int_0^R\left\{A_r^2+\dfrac{l^2}{R^2}A^2\right\}rdr\geq C\int_0^R\left\{A_r^2+A^2\right\}rdr,
\end{align}
it follows that $H$ is an embedded subspace of the standard Sobolev space defined over the disc of radius $R$ centered at the origin, $W^{1,2}_0(D_R)$, and may be viewed as the space of radially symmetric functions with the property $A(0)=0$ for all $A\in H$. With $H$ defined, we may now consider the product space $\textbf{H}=H\times H$ with the induced component-wise operations of $W^{1,2}_0(D_R)$, inner product, and norm  
\begin{align}
\langle (A_1,A_2),(\tilde{A}_1,\tilde{A}_2)\rangle=\langle A_1,\tilde{A}_1\rangle+\langle A_2,\tilde{A}_2\rangle\quad ||(A_1,A_2)||_{\textbf{H}}^2=||A_1||^2+||A_2||^2,
\end{align}
respectively.

The following two lemmas show that the functional \eqref{IndfAction} has a mountain pass structure and is indefinite. 
\begin{lemma}
    There are constants $K,C_0>0$ such that 
    \begin{align}
        \inf\left\{J(A_1,A_2):||(A_1,A_2)||_\textbf{H}^2=K\right\}\geq C_0.
    \end{align}
\end{lemma}
\begin{proof}
    Let $K>0$ be a constant and $(A_1,A_2)\in\textbf{H}$ such that $||(A_1,A_2)||^2_\textbf{H}=K.$ From \eqref{bndOnA^4}, we get
\begin{equation}
    \int_0^RrA^4dr\leq 2\left(\int_0^RrA^2dr\right)\left(\int_0^R rA_r^2dr\right)^{1/2}\left(\int_0^R\dfrac{1}{r} A^2dr\right)^{1/2}\leq 2R^2K^2.
\end{equation}
So, 
\begin{align}
    J(A_1,A_2)&\geq \dfrac{1}{4}||(A_1,A_2)||^2_\textbf{H}-\int_0^RrA_1^2A_2dr\\
    &\geq \dfrac{1}{4}||(A_1,A_2)||^2_\textbf{H}-\left(\int_0^RrA_1^4dr\right)^{1/2}\left(\int_0^RrA_2^2dr\right)^{1/2}\\
    &\geq\dfrac{1}{4}K-(2R^2K^2)^{1/2}(R^2K)^{1/2}\\
    &=\dfrac{1}{4}K-\sqrt{2}R^2K^{3/2}:=f(K).
\end{align}
Then $f\left(\dfrac{1}{72R^4}\right)=\dfrac{1}{864R^4}$ is the maximum value of $f$. Therefore, we have the lower bound
\begin{align}
    J(A_1,A_2)\geq \dfrac{1}{864R^4},  \quad ||(A_1,A_2)||^2_\textbf{H}=\dfrac{1}{72R^4}, 
\end{align}
as desired. 
$\qquad\square$
\end{proof}

\begin{lemma}
    For any $K>0$ there is an $(A_1,A_2)\in\textbf{H}$ such that 
        $||(A_1,A_2)||^2_\textbf{H}>K$ and $J(A_1,A_2)<0$. Moreover, $J$ is indefinite on $\textbf{H}$.
\end{lemma}
\begin{proof}
Consider the following function
\begin{align}
	A_0(r)=\left\{
	\begin{array}{cc}
		\frac{b}{a}r, &0\leq r\leq a,\\
		\frac{b}{a}(2a-r), &a\leq r\leq 2a,
	\end{array}\right.\label{A0},
\end{align}
where $R=2a$. It can be shown as in \cite{YR} that the function $A_0$ is the limit of a Cauchy sequence in $H$ and, consequently, belongs in $H$.  By direct calculation we get
\begin{align}
    \int_0^{2a}rA_0^2dr&=\dfrac{2}{3}a^2b^2,\\
    \int_0^{2a}rA_{0,r}^2dr&=2b^2,\\
    \int_0^{2a}\dfrac{A_0^2}{r}dr&=2b^2(2\ln(2)-1),\\
    \int_0^{2a}rA_0^3dr&=\dfrac{1}{2}a^2b^3. 
\end{align}
It then follows that
    \begin{align}
        ||(A_0,A_0)||^2_\textbf{H}&=8b^2\ln(2),\\
        J(A_0,A_0)&=b^2\left[\dfrac{3}{2}+3l^2(2\ln(2)-1)+\dfrac{2}{3}(3\kappa+\beta)a^2-\dfrac{1}{2}a^2b\right].
    \end{align}
    Therefore, for any $K>0$, we can choose $b$ large enough so that $8b^2\ln(2)>K$ and $J(A_0,A_0)<0$. Additionally, $J(A_0,A_0)\rightarrow-\infty$ as $b\rightarrow\infty$ and, consequently, $J$ is indefinite.
    $\qquad\square$
\end{proof}

We now prove that $J$ satisfies the Palais-Smale condition. Note that it is straightforward to show that $J$ is $C^1(\textbf{H})$ and recall that a $C^1$-functional $J:\textbf{H}\rightarrow\mathbb{R}$ is said to be Palais-Smale if for any sequence $\left\{(A_{1,n},A_{2,n})\right\}_{n=1}^{\infty}\in\mathbf{H}$ such that  $\left\{J(A_{1,n},A_{2,n})\right\}_{n=1}^{\infty}$is bounded in $\mathbf{H}$ and $J'(A_{1,n},A_{2,n})\rightarrow 0$ as $n\rightarrow\infty$ (as a sequence in the dual of $\mathbf{H}$), implies the existence of a strongly convergent sub-sequence of $\left\{(A_{1,n},A_{2,n})\right\}_{n=1}^{\infty}$ converging to an element $(A_1,A_2)$ in $\mathbf{H}$ \cite{Jabri,Rabi}.
\begin{lemma}
    The action functional $J$ defined by \eqref{IndfAction} is Palais-Smale. 
\end{lemma}
\begin{proof} 
    Let $\{(A_{1,n},A_{2,n})\}$ be a sequence in $\textbf{H}$ such that 
\begin{align}
    J(A_{1,n},A_{2,n})=&\dfrac{1}{2}\int_0^R\left([A_{1,n}]_r^2+\dfrac{1}{2}[A_{2,n}]_r^2+\dfrac{l^2}{r^2}A_{1,n}^2+\dfrac{2l^2}
    {r^2}A_{2,n}^2\right)rdr\label{PS1}\\
   &+ \int_0^R\left(\kappa A_{1,n}^2+(2\kappa+\beta)A_{2,n}^2-A_{1,n}^2A_{2,n}\right)rdr\rightarrow \alpha,\quad n\rightarrow\infty,\nonumber
\end{align}
and 
\begin{align}
    |J'(A_{1,n},A_{2,n})(\tilde{A}_1,\tilde{A}_2)|\leq \delta_n||(\tilde{A}_1,\tilde{A}_2)||_{\textbf{H}}, \quad\delta_n\geq 0, \quad (\tilde{A}_1,\tilde{A}_2)\in\textbf{H},\label{PS2}
\end{align}
where $\delta_n\rightarrow 0$ as $n\rightarrow\infty$. Take $(\tilde{A}_1,\tilde{A}_2)=(A_{1,n},A_{2,n})$ in \eqref{PS2} to arrive at
\begin{align}
    -(J'(A_{1,n},A_{2,n})(A_{1,n},A_{2,n}))\leq \delta_n||(A_{1,n},A_{2,n})||_{\textbf{H}}
\end{align}
or equivalently 
\begin{align}
    \int_0^R A_{1,n}^2A_{2,n}rdr&\leq\dfrac{1}{3}\int_0^R\left([A_{1,n}]_r^2+\dfrac{1}{2}[A_{2,n}]_r^2+\dfrac{l^2}{r^2}A_{1,n}^2+\dfrac{2l^2}
    {r^2}A_{2,n}^2\right)rdr\label{bnd1}\\
    &+\dfrac{2}{3}\int_0^R\left(\kappa A_{1,n}^2+(2\kappa+\beta)A_{2,n}^2\right)rdr+\dfrac{1}{3}\delta_n||(A_{1,n},A_{2,n})||_{\textbf{H}}\nonumber.
\end{align}
Without loss of generality from \eqref{PS1} we can assume $J(A_{1,n},A_{2,n})\leq \alpha+1$ for all $n=1,2,\ldots$. Using \eqref{bnd1}, we then get
\begin{align}
    \alpha+1\geq&\dfrac{1}{6}\int_0^R\left([A_{1,n}]_r^2+\dfrac{1}{2}[A_{2,n}]_r^2+\dfrac{l^2}{r^2}A_{1,n}^2+\dfrac{2l^2}
    {r^2}A_{2,n}^2\right)rdr\\
    &+\dfrac{1}{3}\int_0^R\left(\kappa A_{1,n}^2+(2\kappa+\beta)A_{2,n}^2\right)rdr-\dfrac{1}{3}\delta_n||(A_{1,n},A_{2,n})||_{\textbf{H}}\nonumber\\
    \geq&\dfrac{1}{12}\int_0^R\left([A_{1,n}]_r^2+[A_{2,n}]_r^2+\dfrac{l^2}{r^2}A_{1,n}^2+\dfrac{l^2}
    {r^2}A_{2,n}^2\right)rdr\\
    &-\dfrac{1}{3}\delta_n||(A_{1,n},A_{2,n})||_{\textbf{H}}\nonumber\\
    \geq&\dfrac{1}{48}||(A_{1,n},A_{2,n})||_{\textbf{H}}^2-\dfrac{4}{3}\delta_n,\quad n=1,2,\dots,
\end{align}
where in the last inequality we used Cauchy's inequality with $\epsilon$ \cite{Evans} and, consequently, we get that $\left\{(A_{1,n},A_{2,n})\right\}_{n=1}^{\infty}$is bounded in $\mathbf{H}$. Without loss generality, we may assume that $\left\{(A_{1,n},A_{2,n})\right\}_{n=1}^{\infty}$ converges weakly to an element $(A_1,A_2)\in\mathbf{H}$ as $n\rightarrow\infty$. Moreover, $(A_{1,n},A_{2,n})\rightarrow (A_1,A_2)$ as $n\rightarrow\infty$ strongly in any $L^p(D_R)\times L^p(D_R)$ or, equivalently, $L^p((0,R),rdr)\times L^p((0,R),rdr)$ for any $p\geq 1$. Letting $n\rightarrow\infty$ in \eqref{PS2}, we obtain
\begin{align}
    0=&\int_0^R\left([A_1]_r[\tilde{A}_1]_r+\dfrac{1}{2}[A_2]_r[\tilde{A}_2]_r+\dfrac{l^2}{r^2}A_1\tilde{A}_1+\dfrac{2l^2}{r^2}A_2\tilde{A}_2\right)rdr\label{PS2_mod1}\\
    &+\int_0^R\left(2\kappa A_1\tilde{A}_1+2(2\kappa+\beta)A_2\tilde{A}_2-2A_1\tilde{A}_1A_2-A_1^2\tilde{A}_2\right)rdr\nonumber\quad\forall (\tilde{A}_1,\tilde{A_2})\in\textbf{H}.
\end{align}
Let $(\tilde{A}_1,\tilde{A}_2)=(A_{1,n}-A_1,A_{2,n}-A_2)$ and insert into \eqref{PS2_mod1} and \eqref{PS2} and insert the resulting \eqref{PS2_mod1} into the resulting \eqref{PS2}, to arrive at 
\begin{align}
    &\bigg|\int_0^R\left(\left([A_{1,n}]_r-[A_1]_r\right)^2+\dfrac{l^2}{r^2}(A_{1,n}-A_1)^2\right)rdr\\
    &\int_0^R\left(\dfrac{1}{2}\left([A_{2,n}]_r-[A_2]_r\right)^2+\dfrac{2l^2}{r^2}(A_{2,n}-A_2)^2\right)rdr\nonumber\\
    &+\int_0^R\left(2\kappa(A_{1,n}-A_1)^2+2(2\kappa+\beta)(A_{2,n}-A_2)^2-2(A_{1,n}A_{2,n}-A_1A_2)(A_{1,n}-A_1)\right)rdr\nonumber\\
    &-\int_0^R\left((A_{1,n}^2-A_1^2)(A_{2,n}-A_2)\right)rdr\bigg|\leq\delta_n||(A_{1,n}-A_1,A_{2,n}-A_2)||_{\textbf{H}}\nonumber,
\end{align}
which then gives
\begin{align}
    &\dfrac{1}{2}||(A_{1,n}-A_1,A_{2,n}-A_2)||_{\textbf{H}}^2\leq\delta_n||(A_{1,n}-A_1,A_{2,n}-A_2)||_{\textbf{H}}\\
    &+\int_0^R\left|(A_{1,n}^2-A_1^2)(A_{2,n}-A_2)\right|rdr+2\int_0^R\left|(A_{1,n}A_{2,n}-A_1A_2)(A_{1,n}-A_1)\right|rdr\nonumber. 
\end{align}
Applying Cauchy's inequality with $\epsilon$ \cite{Evans}, we then get
\begin{align}
    &\dfrac{1}{4}||(A_{1,n}-A_1,A_{2,n}-A_2)||_{\textbf{H}}^2\leq\delta_n^2+\int_0^R\left|(A_{1,n}^2-A_1^2)(A_{2,n}-A_2)\right|rdr\\
    &+2\int_0^R\left|(A_{1,n}A_{2,n}-A_1A_2)(A_{1,n}-A_1)\right|rdr,\quad n=1,2,\ldots\nonumber, 
\end{align}
which then gives the strong convergence of $(A_{1,n},A_{2,n})\rightarrow (A_1,A_2)$ as $n\rightarrow\infty$ in $\textbf{H}$, as desired. 
$\qquad\square$
\end{proof}
With the above lemmas, we can now obtain our desire results as a consequence of the classical mountain-pass theorem \cite{Ambro,Evans,Jabri}. First, from Lemma 4.3, we have that the functional $J$, given by \eqref{IndfAction}, satisfies the Palais-Smale condition. Then by Lemma 4.1 and Lemma 4.2, there are constants $K,C_0>0$ and an element $(A_0,A_0)\in\textbf{H}$ such that $||(A_0,A_0)||_{\textbf{H}}^2>K$ and $J(A_0,A_0)<0$. Now consider the set 
    \begin{align}
        \Gamma=\{g\in\mathcal{C}([0,1];\textbf{H})|g(0)=(0,0),g(1)=(A_0,A_0)\}
    \end{align}
     of all continuous paths in $\textbf{H}$ that link the zero element $(0,0)$ of $\textbf{H}$ to $(A_0,A_0)$. Then there is some $t_g\in(0,1)$ such that $||g(t_g)||_{\textbf{H}}^2=K$. It then follows from the classical mountain-pass theorem that
     \begin{align}
         C=\inf\limits_{g\in\Gamma}\max\limits_{t\in[0,1]}J(g(t))\geq C_0
     \end{align}
     is a critical value of $J$, meaning, that there is an element $(A_1,A_2)\in\textbf{H}$ such that $J(A_1,A_2)=C$, which is a critical point of $J$. Moreover, such critical point is not the trivial solution, $(0,0)$, and by Lemma 3.1 it is also not a semi-trivial solution as desired.\\

\begin{acknowledgements}
We thank the anonymous referee for the careful reading of this manuscript and helpful feedback. This work does not have any conflicts of interest. There are no funders to report for this submission.  
\end{acknowledgements}

\end{document}